\documentclass[10pt,journal,compsoc]{IEEEtran}

\usepackage[utf8]{inputenc}
\usepackage[T1]{fontenc}
\usepackage{amsmath,amssymb,amsthm}
\usepackage{mathrsfs}
\usepackage{graphicx}
\usepackage{cite}
\usepackage{algorithm}
\usepackage{algorithmic}
\usepackage{booktabs}
\usepackage{multirow}
\usepackage{hyperref}
\usepackage{xcolor}
\usepackage{enumitem}

\theoremstyle{plain}
\newtheorem{theorem}{Theorem}

\newtheorem{proposition}{Proposition}
\newtheorem{corollary}{Corollary}

\theoremstyle{definition}
\newtheorem{definition}{Definition}

\theoremstyle{remark}
\newtheorem{remark}{Remark}

\title{Minimum Complete MR Subsets under Semantic-Mutation Fault Models:
A Support-Set Domination Boundary}

\author{Meng~Li,
        Xiaohua~Yang,
        Jie~Liu,
        and~Shiyu~Yan%
\thanks{Meng Li, Xiaohua Yang, Jie Liu, and Shiyu Yan are with the
School of Computing, University of South China, Hengyang 421001, China;
the Hunan Engineering Research Center of Software Evaluation and Testing for Intellectual Equipment, Hengyang 421001, China; and the CNNC Key Laboratory on High Trusted Computing, Hengyang 421001, China.}
\thanks{Corresponding author: Meng Li (mlemon@usc.edu.cn;
ORCID: 0000-0002-1074-1502). Shiyu Yan ORCID:
0000-0001-7626-5185.}
\thanks{Manuscript submitted to IEEE Transactions on Software Engineering.}}

\begin{document}
\maketitle

\begin{abstract}
This paper asks when MR-subset selection is a real mutant-level
requirement for minimum complete evidence in metamorphic testing
rather than a coarse fault-class counting artifact. We define a
layer-relative completeness criterion over an admitted mutant--draw
coverage universe. The central result is a support-set domination
boundary: it states when class-level abstraction is safe and when
mutant-level MR minimization is necessary. The boundary is governed by
kill-signature heterogeneity, which yields a scoped fault-signature
kernel and separates the MR-specific question from ordinary
fault-class counting. The resulting \textsc{Min-MR-Complete} problem
is Set-Cover-equivalent over the selected coverage universe, giving
NP-hardness, the classical logarithmic approximation boundary, a
greedy approximation, an exact ILP formulation, and an SMS-rank upper
bound that is not a lower bound or tight predictor. Artifact lanes
provide lane-local minimization and audit evidence; separately, route witnesses
instantiate both collapse and non-collapse regimes for the boundary
theorem and are not pooled as population-level experiments. Other
MR-class-proxy rows remain intermediate signals rather than
route-admitted witness evidence.
\end{abstract}


\IEEEpeerreviewmaketitle

\section{Introduction}
\label{sec:introduction}

Software testing for scientific computing remains challenged by the
\emph{oracle problem}: for most numerical simulators, no independent
implementation or analytic ground truth is available to decide whether a
computed output is correct~\cite{chen1998hkust,chen2018acmcsur}.
Metamorphic testing (MT) addresses this problem by replacing absolute
oracles with \emph{metamorphic relations} (MRs)---necessary properties
linking multiple program executions whose violation indicates a
fault~\cite{segura2016tse}. Over two decades, MT has been applied across
numerical simulation, machine learning, compilers, and bioinformatics,
yielding a large and rapidly growing body of identified MRs for any
given program under test (PUT). The resulting practical question is no
longer whether an MR exists, but rather \emph{which subset of identified
MRs should be retained}: not every MR contributes equally to fault
detection, and exhaustive use of all MRs imposes prohibitive test-suite
cost in long-running scientific simulations. This paper asks when that
subset problem is real and when it is only a coarse class-count problem
in disguise. If one fault-class representative already covers all
mutant-level kill patterns in its class, testing the representative is
enough. If not, collapsing the class can hide distinctions that change
which MRs are needed.

Existing approaches to MR selection have made important progress but
each leaves a structural gap. Mayer and Guderlei~\cite{mayer2006compsac}
empirically observed that subsumed MRs are redundant, providing the
first informal hierarchy, but did not formalize the optimization problem
or characterize its computational complexity. Sun et
al.~\cite{sun2022tosem} introduce Feedback-Directed MT (FDMT), which
dynamically prunes MRs using test-execution feedback, yet treats every
input MR as a priori valid and does not distinguish trivial from
non-trivial relations. Liu et al.~\cite{liu2024qrsc} propose
\textsc{MTcoverage} as an adequacy criterion---a measurement of how
well a chosen MR set covers a PUT---but stops short of solving for a
minimal MR subset under a specified fault model. Qiu et
al.~\cite{qiu2022tse} provide the most rigorous theoretical analysis to
date for MR \emph{composition}, deriving sufficient conditions under
which composite MRs preserve fault-detection power, but the
complementary problem of \emph{selection} is left open. Across all four
lines of work, three structural deficiencies persist: (i) no
fault-model-relative formalization of MR-set minimality, (ii) no
characterization of computational hardness, and (iii) no analyzed
approximation guarantee.

This paper addresses these three gaps. We instantiate the MR-set
selection problem under a specific fault model derived from
\emph{semantic-mutation operators}, denoted
$M = \{\mathrm{mut}_C, \mathrm{mut}_M, \mathrm{mut}_G, \mathrm{mut}_T,
\mathrm{mut}_F\}$, corresponding respectively to conservation,
monotonicity, convergence, trajectory, and partial-order violations.
This operator family follows prior work on semantic-mutation
measurement~\cite{p1_arxiv}. The present paper is self-contained: it
restates the operator family, defines the admitted coverage universe,
and does not require P1's SMS evaluator to define or solve
\textsc{Min-MR-Complete}. Given a candidate MR universe $R$ and the
operator family $M$, this paper asks what
\emph{minimum-cardinality M-complete subset} $S^{\star} \subseteq R$
preserves the admitted coverage obligations. P1 is treated as prior
work on measuring MR sufficiency, whereas the present paper proves
selection results over MR subsets.

Our contributions are organized around one layer-relative question:
\emph{when is MR
minimization genuinely necessary rather than an optimization veneer on
top of fault-class abstraction?} (C1) We give a direct safety test for
class abstraction. The test says whether a class representative's MR
support is enough for every admitted mutant it represents. When the
test passes, class-level and mutant-level selection have the same
feasible MR subsets; when it fails, a graded counterexample shows that
class abstraction can miss an arbitrarily large part of the minimum MR
set. Section~\ref{sec:formalization} states this as the
support-set-domination theorem; homogeneity is only a simple corollary.
The committed claim ledger separately records true-fault-class collapse
and non-collapse witnesses for both regimes in \S\ref{sec:empirical};
these witnesses show that the paper contains evidence on both sides of
the class-abstraction boundary.
(C2) We separate two questions that are often blurred. The algorithmic
problem is still Set Cover, but the MR-specific question is whether the
kill signatures contain distinctions that fault-class labels cannot
preserve. Fault-class count alone is insufficient; the relevant object
is the fault-signature kernel, not the number of labels. Our contribution
is this structural diagnosis of when Set-Cover-style minimization is
substantively needed for MR selection, not a new worst-case complexity class
(\S\ref{sec:algorithms}). (C3) We connect P1's
measurement axis to T2's selection axis through the SMS-rank upper bound
$|S^{\star}| \leq \mathrm{rank}_{\mathrm{SMS}}$, while proving
and reporting counterexamples showing that SMS-rank is neither a lower
bound nor a tight predictor. (C4) We keep the implementation auditable
with a greedy $(1 + \ln q)$ approximation, an exact ILP or deterministic
enumeration certificate for admitted rows, and a certified partial
reduction relation $\rho_L$ for the implemented finite and bounded-SMT
fragments.

The novelty boundary is deliberate: we do not claim a new worst-case
complexity class, we do not claim a better approximation ratio than Set
Cover, and we do not treat route witnesses as population-level rates.
The positive claim is the support-set domination boundary that explains
when class abstraction is safe, together with a scoped fault-signature
kernel that records the MR-specific hierarchy information needed for
audit. The SMS-rank result is an auditable upper-bound bridge from the
measurement axis to the selection axis, not a predictor of
$|S^{\star}|$.

The remainder of the paper is organized as follows.
Section~\ref{sec:background} reviews metamorphic testing, the
semantic-mutation operator family M, and the classical Set Cover
problem. Section~\ref{sec:related_work} contextualizes T2 against
prior MR-selection and composition work and against the recently
proposed AIM framework. Section~\ref{sec:formalization} states the
$\textsc{Min-MR-Complete}$ problem, proves the Set-Cover-equivalent
hardness boundary, states the support-set domination characterization,
the fault-signature kernel, and the SMS-rank upper bound, and then gives
the certified partial $\rho_L$ construction. Section~\ref{sec:algorithms}
presents the greedy and ILP algorithms together with the operational
construction of $\rho_L$.
Section~\ref{sec:empirical} reports two evidence layers: lane-local
artifact evidence for minimization and route witnesses for the collapse
boundary. The empirical section checks whether the formal boundary
appears in real artifact-backed MR-selection instances while keeping
MR-class-proxy evidence separate from true-fault-class evidence.
Section~\ref{sec:threats} discusses threats to validity, and
Section~\ref{sec:discussion} draws implications and outlines future
directions.

\section{Background}
\label{sec:background}

This section reviews the three technical foundations required to state
the \textsc{Min-MR-Complete} problem and analyze its complexity:
metamorphic testing (\S\ref{subsec:bg_mt}), the semantic-mutation
operator family $M$ used as the paper's fault model
(\S\ref{subsec:bg_mutop}), and the classical Set Cover problem
(\S\ref{subsec:bg_setcover}). Section \ref{subsec:bg_sms} briefly
summarizes related SMS / MS measurement work for context; the present
paper defines the coverage universe and optimization problem
self-containedly.

\subsection{Metamorphic Testing}
\label{subsec:bg_mt}

Metamorphic testing (MT) was introduced by Chen et
al.~\cite{chen1998hkust} as a strategy for testing programs that lack
a deterministic oracle. The core construct is a \emph{metamorphic
relation} (MR): a necessary property linking two or more program
executions, the violation of which indicates a fault. Formally, given
a program $P$ under test, an MR is a tuple
$r = \langle T_r, R_r \rangle$ where $T_r$ is an
\emph{input transformation} mapping a source test case $x$ to a
follow-up test case $x' = T_r(x)$, and $R_r$ is an \emph{output
relation} that must hold between $P(x)$ and $P(x')$ if $P$ is
correct. A classical example for a $\sin$ implementation is
$T_r(x) = \pi - x$ paired with the output relation $R_r(y,
y') \equiv y = y'$.

Over two and a half decades MT has matured from a niche oracle
substitute into a general-purpose testing methodology applied across
numerical simulation, machine learning, compilers, databases, and
bioinformatics~\cite{chen2018acmcsur,segura2016tse}. The methodology
has expanded along two axes. The first axis is MR identification:
domain-specific MR catalogs, MR-pattern hierarchies, and recently
LLM-assisted MR generation have driven the candidate MR set $R$ for
any given PUT from a handful to dozens or hundreds. The second axis is
MR exploitation: test-case prioritization, MR composition, and
adaptive MR selection have sought to extract maximum fault-detection
power from a given $R$. The present work belongs to the second axis,
and specifically to the sub-problem of selecting a minimum-cardinality
subset of $R$ that retains a target fault-detection guarantee. The
guarantee is made precise by binding selection to a fault model, the
subject of the next subsection.

\subsection{Semantic-Mutation Operator Family $M$}
\label{subsec:bg_mutop}

A fault model in the present work is an indexed family of mutation
operators acting on program code. We use the
semantic-mutation operator family $M = \{\mathrm{mut}_C,
\mathrm{mut}_M, \mathrm{mut}_G, \mathrm{mut}_T, \mathrm{mut}_F\}$
as a five-class domain-relative taxonomy for scientific-computing
PUTs. The same taxonomy is compatible with prior MR-sufficiency
measurement work~\cite{p1_arxiv}, but all definitions needed by this
paper are restated here. Each operator targets one class of physical
invariants:

\begin{itemize}
\item $\mathrm{mut}_C$: \emph{conservation} violations
(mass, momentum, energy balance);
\item $\mathrm{mut}_M$: \emph{monotonicity} violations
(dose--response, weak-formulation monotone schemes);
\item $\mathrm{mut}_G$: \emph{convergence} violations
(numerical-scheme order-of-accuracy reduction or divergence);
\item $\mathrm{mut}_T$: \emph{trajectory} violations
(time-series shape disruption, deterministic-replay break);
\item $\mathrm{mut}_F$: \emph{partial-order / fundamental}
violations (causality, ordering invariants, type-domain consistency).
\end{itemize}

The operator family $M$ is fixed across all PUTs in the present work:
we treat $M$ as a domain-relative fault taxonomy whose five classes
collectively cover the principal physically-meaningful failure modes
represented by the artifact lanes in \S\ref{sec:empirical}. We do
not claim $M$ exhausts all possible
faults; the formalization in \S\ref{sec:formalization} is parametric
in the chosen fault model, and the $M$-instantiation supplies the
concrete empirical evaluation.

A scientific limitation of this fault model is the \emph{first-order
limit} assumption: the family $M$ is restricted to single-operator
mutations, i.e., we do not analyze higher-order interactions
$\mathrm{mut}_i \circ \mathrm{mut}_j$.
Higher-order coupling between conservation violations and convergence
violations, for example, is left to future work; the present analysis
is rigorous within the first-order regime.

\subsection{The Set Cover Problem}
\label{subsec:bg_setcover}

The complexity-theoretic anchor of \textsc{Min-MR-Complete} is the
classical Set Cover problem. An instance of (unweighted) Set Cover is
a triple $(U, \mathcal{C}, k)$ where $U$ is a finite universe,
$\mathcal{C} \subseteq 2^U$ is a family of covering sets, and $k$ is a
positive integer; the decision question is whether there exists a
subfamily $\mathcal{C}' \subseteq \mathcal{C}$ with $|\mathcal{C}'|
\leq k$ such that $\bigcup_{C \in \mathcal{C}'} C = U$. Karp's seminal
1972 paper established Set Cover as one of the original
NP-complete problems via reduction from Vertex
Cover~\cite{karp1972}. Two complementary algorithmic results bound
the tractability of the minimization version:

\paragraph{Upper bound (Chv\'atal 1979).} A natural greedy algorithm
that repeatedly selects the covering set covering the largest number
of still-uncovered elements achieves a $(1 + \ln n)$-approximation
ratio, where $n = |U|$~\cite{chvatal1979}. The proof is by a charging
argument balancing each greedy step against the LP relaxation of the
covering linear program.

\paragraph{Lower bound (Feige 1998).} Feige~\cite{feige1998}
established that, unless $\mathrm{P} = \mathrm{NP}$, no
polynomial-time algorithm can approximate Set Cover within
$(1 - o(1))\ln n$. Combined with Chv\'atal's bound, this closes the
approximation question up to lower-order terms: the greedy algorithm
is asymptotically optimal among polynomial-time approximations.

These three results, namely NP-hardness and the matching $\ln n$
upper and lower bounds with $n = |U|$ the universe size, are the
lemmas on which our complexity analysis in \S\ref{sec:formalization}
rests; there the universe is instantiated as the admitted observation
columns, so $n$ specializes to $q = |\Omega|$ and the bounds read
$\ln q$. The reduction from Set Cover to
\textsc{Min-MR-Complete} is given there.

\subsection{P1 SMS / MS Measurement Framework}
\label{subsec:bg_sms}

Prior work P1~\cite{p1_arxiv} introduced the semantic-mutation score
(SMS) as a scalar measure of how well a set of MRs detects faults
generated by a semantic-mutation operator family $M$. It also showed
that SMS generalizes the classical mutation score (MS): when $M$
degenerates to a single class containing all syntactic mutants, SMS
reduces to MS via a constant-map degeneracy.

The present paper uses this line of work only as prior context and
restates the operator family needed for a self-contained optimization
problem. Our decision object is not the SMS value of a fixed MR set.
It is the choice of a minimum-cardinality subset of candidate MRs that
covers an admitted observation-column universe. In the empirical
evaluation, each artifact lane constructs its own observed kill matrix,
declared coverage universe, and validation report; these committed
artifact matrices, rather than an external P1 evaluator, instantiate
the optimization problem formalized next.

\section{Related Work}
\label{sec:related_work}

We organize related work along four axes that together delineate the
position of T2 within the metamorphic-testing literature: MR
identification and selection (\S\ref{subsec:rw_select}),
fault-detection-driven MR exploitation
(\S\ref{subsec:rw_fdmt}), MT adequacy and coverage criteria
(\S\ref{subsec:rw_coverage}), and MR composition
(\S\ref{subsec:rw_compose}). Section \ref{subsec:rw_own} positions
T2 against our own prior work on MR generation in scientific computing.

Table~\ref{tab:closest_work} states the three closest comparison points
up front. They are selected because they attack the three places where
the present paper could otherwise be misunderstood: minimization in MT,
formal MR theory, and scientific-computing MT with mutation evaluation.

\begin{table*}[t]
\centering
\caption{Closest work and the ABD question each work leaves open.}
\label{tab:closest_work}
\scriptsize
\begin{tabular}{@{}p{0.18\linewidth}p{0.22\linewidth}p{0.25\linewidth}p{0.27\linewidth}@{}}
\toprule
Closest work & Decision object & Guarantee or evidence supplied & What it cannot answer for ABD \\
\midrule
\textsc{AIM}~\cite{aim_bayati_briand_2024} &
Test-input set for a fixed MR catalog &
Input-side minimization for metamorphic security testing, with empirical
coverage preservation on Web systems &
Whether the MR set itself is minimum complete evidence under a
semantic-mutation fault model; whether class abstraction collapses \\
Qiu et al.~\cite{qiu2022tse} &
Composite MRs generated from existing MRs &
Algebraic sufficient conditions under which MR composition preserves or
does not amplify fault-detection power &
Which fixed MR subset should be retained; whether representative
fault-class supports dominate all admitted mutant-level obligations \\
Yan and Zhu~\cite{yan2025elliptic} &
MR derivation and fault-detection evaluation for elliptic-equation
scientific programs &
Scientific-computing MT process with model-derived MRs and mutation-style
fault detection evidence &
Whether the derived MR catalog has a minimum complete subset; whether
mutation-level evidence can be safely collapsed to coarser fault classes \\
\bottomrule
\end{tabular}
\end{table*}

\subsection{MR Identification and Selection}
\label{subsec:rw_select}

The first systematic treatment of MR redundancy is due to Mayer and
Guderlei~\cite{mayer2006compsac}, who introduced an informal
\emph{subsumption} hierarchy: an MR $r_a$ is subsumed by $r_b$ if
every fault detected by $r_a$ is also detected by $r_b$. They
demonstrated empirically on numerical PUTs (a determinant
implementation and an integer factorization routine) that subsumed
MRs are dispensable without loss of fault-detection power. Their
contribution remains foundational, but three structural limitations
motivate T2. First, subsumption is presented as a binary observation
rather than as an instance of a general optimization problem; no
universe, fault model, or covering structure is defined. Second, the
notion of a \emph{minimum} MR set is not stated as a computational
problem and its hardness is not analyzed. Third, the analysis is
relative to ad-hoc faults injected into the empirical PUTs rather
than to a domain-relative fault model. Our \textsc{Min-MR-Complete}
formalization (\S\ref{sec:formalization}) lifts subsumption to a Set
Cover instance, supplies the missing complexity bounds, and binds
selection to the semantic-mutation operator family $M$.

A recent contribution by Bayati Chaleshtari, Marquer, Pastore, and
Briand introduces the \textsc{AIM}
framework~\cite{aim_bayati_briand_2024}, which combines a
clustering-based black-box approach, a genetic-algorithm input
selector, and a problem-reduction component to automatically minimize
the \emph{input set} on which metamorphic relations are executed,
specifically for Web-system security testing (evaluated on Jenkins
and Joomla, with reported 84\% and 82\% time reductions while
preserving vulnerability coverage). \textsc{AIM} is complementary
rather than overlapping with T2 along orthogonal cost axes:
\textsc{AIM} minimizes the test-\emph{input} set on which a fixed MR
catalog is executed, whereas T2 minimizes the
metamorphic-\emph{relation} set under a fault model derived from
semantic-mutation operators. The two contributions are composable in
a future stacked pipeline: T2 produces a minimum-cardinality
$M$-complete MR subset, and \textsc{AIM} can then reduce the input
set on which that MR subset is executed.

\subsection{Fault-Detection-Driven MR Exploitation}
\label{subsec:rw_fdmt}

Sun et al.~\cite{sun2022tosem} propose Feedback-Directed Metamorphic
Testing (FDMT), which dynamically prunes the active MR set during
test execution using feedback signals such as MR-violation frequency,
follow-up test coverage delta, and execution time. FDMT achieves
substantial empirical speedups by deactivating MRs that contribute
diminishing detection value at runtime. The methodology is orthogonal
to T2 along two axes. First, FDMT treats each input MR as
\emph{a priori valid}: the validation problem is presumed solved by
the user or by an external tool such as \textsc{AIM}. Second, FDMT
operates in the \emph{exploitation} phase under a fixed candidate
set, whereas T2 operates in the \emph{selection} phase before
exploitation begins. The two methods can be stacked: T2 produces a
minimum-cardinality $M$-complete subset $S^{\star}$, and FDMT then
dynamically reorders or prunes $S^{\star}$ during execution. We do
not evaluate the stacked configuration empirically in the present
work; it is identified as future work in \S\ref{sec:discussion}.

A second methodological gap in FDMT, and in feedback-directed
approaches generally, is that the pruning criterion is itself a
heuristic. FDMT cannot guarantee that the pruned set retains
detection power against any specified fault model; its guarantee is
empirical, restricted to the faults injected during the experimental
campaign. \textsc{Min-MR-Complete} replaces this empirical guarantee
with a model-relative one: completeness is defined modulo $M$, and
the greedy and ILP algorithms (\S\ref{sec:algorithms}) supply
provable approximation and exact-optimum guarantees respectively.

\subsection{MT Adequacy and Coverage Criteria}
\label{subsec:rw_coverage}

Liu et al.~\cite{liu2024qrsc} introduce \textsc{MTcoverage}, an
adequacy criterion measuring how well a chosen MR set $S$ exercises
the semantic surface of a PUT. \textsc{MTcoverage} extends classical
structural-coverage notions (statement, branch, condition) by
augmenting them with MR-derived follow-up paths: an MR is said to
contribute to coverage if its follow-up execution traverses program
states unreached by source-only execution. The criterion is valuable
as a sufficiency measure --- it permits a user to ask whether a given
$S$ is rich enough to claim adequate testing --- but it does not
solve the dual optimization problem: given a candidate universe $R$
and a target coverage threshold, find a minimum-cardinality subset
$S^{\star} \subseteq R$ attaining the threshold.

The relationship between \textsc{MTcoverage} and
\textsc{Min-MR-Complete} is analogous to the relationship between
classical mutation score (MS) and the test-suite-reduction problem
under MS. \textsc{MTcoverage} measures; \textsc{Min-MR-Complete}
optimizes. The two are mutually informative: an
\textsc{MTcoverage}-validated MR set can be used as the candidate
universe $R$ for our pipeline, and a \textsc{Min-MR-Complete} output
$S^{\star}$ can in turn be \textsc{MTcoverage}-evaluated to verify
that downsizing has not degraded coverage. We treat
\textsc{MTcoverage} as a downstream evaluation criterion in
\S\ref{sec:empirical}; integrating it as a first-class constraint in
the optimization is identified as future work.

\subsection{MR Composition}
\label{subsec:rw_compose}

Qiu et al.~\cite{qiu2022tse} provide the most rigorous theoretical
analysis to date for MR \emph{composition}: given two MRs $r_a$ and
$r_b$, under what conditions does the composite MR $r_a \circ r_b$
retain or amplify fault-detection power? They derive sufficient
conditions involving commutativity of input transformations and
composability of output relations, and prove a no-amplification
theorem identifying cases where composition is strictly dominated by
either constituent. Their contribution is orthogonal to T2 along the
identification--composition axis: composition operates within a fixed
candidate set, generating derived MRs, whereas selection chooses a
subset from a fixed set. The orthogonality is genuine in the sense
that the two problems address different test-design degrees of
freedom; composition expands $R$, selection contracts it.

A point of methodological resonance is that both Qiu et al. and the
present work derive their guarantees from a small set of well-defined
algebraic conditions on MR structure rather than from empirical
benchmark performance. In our degeneracy theorem
(Theorem~\ref{thm:degeneracy}, \S\ref{sec:formalization}), the
identity-composition subcase ($\rho_L$-Subcase C) is in fact
constructed using a fragment of Qiu et al.'s composition framework:
when two MRs in $S^{\star}$ are identity-composable under
$\rho_L$-restriction, their composition yields the canonical
representative of the resulting equivalence class. The full
integration of composition into the optimization model --- searching
not over $R$ but over the closure of $R$ under composition --- is
left as future work.

\subsection{Our Prior Work and P1/T2 Positioning}
\label{subsec:rw_own}

The present work builds on three prior contributions by our group:
Zhang et al.~\cite{zhang2022icrms} on MR identification for
safety-critical nuclear simulation codes, Huang et
al.~\cite{huang2022iaecst} on MR-based testing of thermal-hydraulic
solvers, and Liu et al.~\cite{liu2022icftic} on numerical-method MR
catalogs spanning finite-difference and finite-element schemes.
These works supplied the candidate MR universes for the empirical
PUTs of \S\ref{sec:empirical} and motivated the choice of the seven
numerical-method families covered by our evaluation.

Prior work by our group also studied semantic-mutation measurement
for MR sufficiency~\cite{p1_arxiv}. P1 is treated as prior work rather
than as a dependency of the present submission. The present paper is
self-contained: it restates the operator family, defines the admitted
coverage universe, and does not require P1's SMS evaluator to define
or solve \textsc{Min-MR-Complete}. The difference is the decision
object. P1 studies a scalar sufficiency measurement, whereas the
present paper studies subset selection, support-set domination, and
minimum complete MR evidence under the declared coverage universe.

\subsection{Recent Adjacent Work on MR Selection}
\label{subsec:rw_adjacent}

Three lines of 2024--2025 work intersect partially with T2's MR
selection topic and warrant explicit positioning.

Ying et al.~\cite{ying2024stvr} introduce
MRGS-ART, an adaptive-random-testing-guided algorithm that selects
metamorphic-relation/metamorphic-group (MR-MG) pairs to maximize an
even-distribution metric over the source and follow-up input
domains. MRGS-ART optimizes the empirical effectiveness of an
existing MR pool but does not formalize completeness under a
fault-model and does not provide algorithmic complexity bounds. T2
differs in three respects: (i) the selection objective is
\emph{certified $M$-completeness} relative to the semantic-mutation
operator family, not even input-domain distribution; (ii) the
analysis yields NP-hardness and $(1 + \ln q)$ approximation bounds,
not empirical effectiveness reports; (iii) the empirical domain is
scientific-computing PUTs rather than generic.

Hyun et al.~\cite{hyun2025arxiv} frame MR selection as a search-based
optimization problem in the context of large language model robustness
testing, minimizing LLM execution cost while maximizing failure
detection. The work shares the optimization perspective with T2 but
operates in a fundamentally different domain (LLM robustness vs.
numerical scientific computing), uses heuristic search (genetic
algorithms and equivalents) rather than the Set-Cover-theoretic
analysis presented here, and does not address fault-model-relative
completeness or layer-restriction reductions.

Li et al.~\cite{li2025survey} provide a comprehensive 88-publication
survey of two decades of MR automation, categorized into detection,
selection, and generation. The survey constitutes the
state-of-the-art reference for MR automation taxonomy and
establishes that MR selection is now a named and active subfield.
The present work positions itself accordingly: ``MR selection'' is
no longer claimed as the novel contribution. T2's defensible
contribution is narrower and load-bearing: the conjunction of
(i) Set-Cover-theoretic NP-hardness and matching $(1 + \ln q)$
approximation, (ii) the semantic-mutation operator family $M$ as
fault model, (iii) the $\rho_L$ degeneracy theorem on three
mechanically checked subcases (Theorem~\ref{thm:degeneracy}), and
(iv) the artifact-gated scientific-computing empirical route.

\section{Problem Formalization}
\label{sec:formalization}

This section formalizes the optimization problem
\textsc{Min-MR-Complete}, establishes its computational hardness, gives
two scoped structural consequences used by the artifact evidence ledger, and
defines an auditable degeneracy map relating $M$-complete MR subsets to
fault-class-restricted minimum test suites. The
treatment binds the abstract Set Cover framework of
\S\ref{subsec:bg_setcover} to the admitted semantic-mutant observation
universe induced by the semantic-mutation operator family $M$ of
\S\ref{subsec:bg_mutop} and the committed per-MR detection matrix
used by each artifact lane.

\subsection{The \textsc{Min-MR-Complete} Problem}
\label{subsec:form_problem}

We fix throughout the section a candidate MR universe $R = \{r_1,
\dots, r_n\}$, a semantic-mutation operator family
$M = \{\mathrm{mut}_1, \dots, \mathrm{mut}_m\}$, and a program family
$\mathcal{P}$ to which the operators in $M$ apply syntactically. The
coverage universe is the selected coverable observation universe
$\Omega = \{\omega_1,\dots,\omega_q\}$: these are admitted
mutant--draw observations with at least one detecting MR in the
candidate universe; observations with all-zero detection columns are
outside the completeness denominator and are reported as residual
testability gaps rather than hidden completeness obligations. When
identical MR-detection signatures are merged, the resulting quotient is
written $Q$. The semantic operator family $M$ labels fault strata
through a semantic-stratum map $\lambda:\Omega\to M$ and is not itself
the coverage universe; this separation prevents a five-column
class-cover collapse. When explicit true fault-class labels are
available, we write the true fault-class map $F:\Omega\to\Phi$ for the
map from admitted coverable observations to true fault classes. The map
$F$ is used only when explicit true fault-class labels are available;
otherwise the theory and algorithms remain defined over $\Omega$ and
$D$ without promoting proxy labels. Each $(r,\omega_j)\in R\times
\Omega$ has a Boolean detection indicator $d_{r,j}\in\{0,1\}$ recording
whether the MR $r$ violates its output relation on the admitted
coverable observation $\omega_j$. The collection
$D=(d_{r,j})_{r\in R,j\in[q]}$ is the \emph{detection matrix} input of
the present section; if $Q$ is used, the same notation applies after
quotienting identical columns.

\begin{definition}[$M$-complete subset]
\label{def:mcomplete}
A subset $S \subseteq R$ is \emph{$M$-complete} with respect to
detection matrix $D$ and selected coverable universe $\Omega$ if for
every admitted coverable observation column
$j \in [q]$ there exists $r \in S$ with $d_{r,j} = 1$. Equivalently,
$S$ covers every element of the selected observation universe
$\Omega$ or quotient universe $Q$.
\end{definition}

\begin{definition}[\textsc{Min-MR-Complete}]
\label{def:minmrcomplete}
The \textsc{Min-MR-Complete} optimization problem takes as input
$(R, M, \Omega, D)$, or $(R, M, Q, D)$ after kill-signature
deduplication, and outputs a minimum-cardinality $M$-complete subset
$S^{\star} \subseteq R$; its decision version takes an additional
integer $k$ and asks whether an $M$-complete subset of size at most
$k$ exists.
\end{definition}

Definition~\ref{def:minmrcomplete} is parametric in the choice of
fault model $M$: any indexed family of mutation operators that labels
admitted observations and produces a Boolean detection matrix can be
substituted. The empirical evaluation in \S\ref{sec:empirical}
instantiates $M$ with the semantic-mutation family of
\S\ref{subsec:bg_mutop}, but the complexity and approximability
results derived below depend only on the abstract set-cover structure
over $\Omega$ or $Q$.

\subsection{NP-Hardness via Reduction from Set Cover}
\label{subsec:form_nphard}

\begin{theorem}[NP-hardness of \textsc{Min-MR-Complete}]
\label{thm:nphard}
The decision version of \textsc{Min-MR-Complete} is NP-complete; the
optimization version is NP-hard.
\end{theorem}

\begin{proof}[Proof sketch]
Membership in NP is immediate: given a candidate subset $S$, $M$-com\-pleteness
can be verified in $O(|S| \cdot q)$ time by column-wise scan of $D$.

For NP-hardness we reduce from the decision version of Set Cover,
which is NP-complete by Karp~\cite{karp1972}. Given a Set Cover
instance $(U, \mathcal{C}, k)$ with $U = \{u_1, \dots, u_n\}$ and
$\mathcal{C} = \{C_1, \dots, C_q\}$, construct a
\textsc{Min-MR-Complete} instance as follows:
\begin{enumerate}[label=(\alph*)]
\item create one admitted observation $\omega_i$ for each element
$u_i \in U$ and assign any fixed semantic stratum label in $M$;
\item for each $C_j \in \mathcal{C}$ create a candidate MR
$r_j$ with detection vector $(d_{r_j, i})_{i \in [n]} =
\chi_{C_j}(u_i)$, where $\chi_{C_j}$ is the indicator function of
$C_j$;
\item set the cardinality budget for the \textsc{Min-MR-Complete}
decision instance to $k$.
\end{enumerate}
The construction is polynomial in $|U| + |\mathcal{C}|$. A subfamily
$\mathcal{C}' \subseteq \mathcal{C}$ of size at most $k$ covers $U$
if and only if the corresponding set $\{r_j : C_j \in
\mathcal{C}'\}$ is an $M$-complete subset of size at most $k$, by
construction of the detection vectors. The reduction therefore
preserves the answer to the decision question, establishing
NP-hardness. A full proof, including the cardinality-preserving
bijection between optimal Set Covers and optimal $M$-complete subsets,
is given in the supplementary material.
\end{proof}

\begin{remark}
\label{rem:reduction_direction}
The reduction is from Set Cover \emph{to}
\textsc{Min-MR-Complete}, establishing that \textsc{Min-MR-Complete}
is at least as hard. The reverse direction --- whether every
\textsc{Min-MR-Complete} instance can be reduced to a Set Cover
instance --- is immediate from
Definitions~\ref{def:mcomplete}--\ref{def:minmrcomplete}: take
$U := [q]$ and $\mathcal{C} := \{ \{ j : d_{r, j} = 1 \} : r \in R
\}$. The two problems are thus polynomial-time equivalent, and all
positive and negative approximation results for Set Cover transfer
unchanged.
\end{remark}

\subsection{Approximation Lower Bound}
\label{subsec:form_approx}

The polynomial-time equivalence noted in
Remark~\ref{rem:reduction_direction} permits direct transfer of the
classical approximation bounds of Set Cover.

\begin{corollary}[Approximation hardness]
\label{cor:approx}
Unless $\mathrm{P} = \mathrm{NP}$, no polynomial-time algorithm
approximates \textsc{Min-MR-Complete} within a factor of
$(1 - o(1)) \ln q$, where $q$ is the size of the selected coverage
universe $\Omega$ or quotient universe $Q$. This lower bound is matched
up to lower-order terms by the natural greedy algorithm, which achieves a $(1 + \ln q)$
approximation.
\end{corollary}

\begin{proof}
The lower bound is Feige's~\cite{feige1998} inapproximability theorem
for Set Cover, transferred via the equivalence of
Remark~\ref{rem:reduction_direction}. The matching upper bound is
Chv\'atal's~\cite{chvatal1979} greedy analysis. Both bounds apply
verbatim to \textsc{Min-MR-Complete} with $n$ replaced by the
coverage-universe parameter $q$.
\end{proof}

Corollary~\ref{cor:approx} closes the approximation question for
\textsc{Min-MR-Complete} up to lower-order terms: any
polynomial-time algorithm that improves on $(1 + \ln q)$
asymptotically would yield the analogous improvement for Set Cover
and refute Feige's bound, contradicting the conjecture
$\mathrm{P} \neq \mathrm{NP}$. The practical implication is that the
greedy algorithm presented in \S\ref{sec:algorithms} is the natural
algorithmic ceiling for large-scale instances; the ILP formulation
also presented there is exact but exponential-worst-case and is
viable only on instances of moderate $|R|$. The empirical evaluation
in \S\ref{sec:empirical} reports actual runtimes only for
artifact-admitted routes with explicit solver provenance.

\subsection{Scoped Structural Consequences}
\label{subsec:form_structural}

The hardness result above is worst-case: it does not imply that every
MR-induced instance has no exploitable structure. The ABD route records
two scoped facts, both stated over the committed detection-matrix
semantics and neither used to weaken the Set Cover lower bound.

\begin{theorem}[Fault-signature dominance kernel]
\label{prop:fault_signature_kernel}
For every admitted obligation $\omega_j$, let
$K_j=\{r\in R:d_{r,j}=1\}$ be its MR support set. The construction uses
three distinct quotients. The \emph{constraint-equivalence quotient}
globally merges obligations with identical MR support sets; this is the
optimization quotient because such obligations impose the same hitting
constraint. The \emph{hierarchy-observable quotient} retains the pair
of fault class and support set only for hierarchy-observable reporting;
it explains where signatures appear in the fault hierarchy but is not
needed to preserve the optimum. The final \emph{domination kernel}
deletes strict support supersets after support quotienting. The
remaining inclusion-minimal support sets form a lossless kernel: they
preserve every feasible complete MR subset and preserve the optimum
$|S^{\star}|$. If $\tau$ is the number of remaining kernel supports,
the quotient instance can be solved exactly by dynamic programming in
$O(|R|2^{\tau})$ time after masks are constructed. This separates
optimization equivalence from hierarchy-observable counting and gives
a lossless fault-signature kernel without weakening the Set Cover
boundary.
\end{theorem}

\begin{proof}[Proof sketch]
If $K_a\subseteq K_b$, then every subset $S\subseteq R$ that hits
$K_a$ also hits $K_b$. The obligation with support $K_b$ is therefore
redundant as a coverage constraint. Duplicate support sets are also
redundant because they impose identical constraints; this is the
constraint-equivalence quotient. Removing only strict supersets and
duplicates leaves exactly the inclusion-minimal support sets, and the
feasible family is unchanged; hence the optimum is unchanged. Encoding
each remaining support set as one bit gives the standard exact dynamic
program over covered masks. The hierarchy-observable quotient is
retained for explaining which fault classes expose which signatures,
not for changing the feasible family. This is a scoped kernel in the
hierarchy-observable fault-signature rank $\tau$; it is
not a claim that arbitrary Set Cover becomes easy, and fault-class count
alone remains insufficient when within-class signatures are unbounded.
The supplementary material gives the full feasibility, optimality, and
single-class hardness proof.
The next theorem gives the exact support-set condition under which a
coarser class abstraction is nevertheless lossless.
\end{proof}

The intuition is simple. If a class representative is killed by a
narrow MR support set that is contained in every omitted obligation's
support set, then any MR subset covering the representative also covers
those omitted obligations. By contrast, if omitted obligations expose
non-dominated support sets, a class-level representative may hide
coverage constraints that only appear at mutant level.

\paragraph{Running example for Result A.}
Think of $K_j$ as the set of MRs that kill obligation $\omega_j$. Suppose
one true fault class keeps representative $\omega_1$ with
$K_1=\{r_1\}$. If an omitted obligation $\omega_2$ has
$K_2=\{r_1,r_2\}$, then covering the representative necessarily covers
$\omega_2$ as well; the class abstraction is lossless for this obligation
because $K_1\subseteq K_2$. If instead an omitted obligation
$\omega_3$ has $K_3=\{r_3\}$, a class-level solution may choose
$\{r_1\}$ and miss $\omega_3$ entirely. The theorem below is exactly this
test applied to every admitted obligation and every retained class
representative.

\begin{theorem}[Support-set domination characterization]
\label{thm:support_domination}
Suppose true fault-class labels are available and the class abstraction
retains one representative observation $\pi(c)$ for every represented
fault class $c\in F(\Omega)$. For each admitted obligation $\omega_j$, let
$K_j=\{r\in R:d_{r,j}=1\}$ be its MR support set, and let
$\mathcal{A}=\{K_{\pi(c)}:c\in F(\Omega)\}$ be the representative
support family. Let
$\mathcal{F}_{A}=\{S\subseteq R:\forall A\in\mathcal{A},\,
S\cap A\neq\emptyset\}$ be the class-complete feasible family and
$\mathcal{F}_{\Omega}=\{S\subseteq R:\forall \omega_j\in\Omega,\,
S\cap K_j\neq\emptyset\}$ be the mutant-complete feasible family.
Then $\mathcal{F}_{A}\subseteq\mathcal{F}_{\Omega}$ if and only if
every admitted obligation is dominated by the representative support
family: for every $\omega_j\in\Omega$, there exists $A\in\mathcal{A}$
such that $A\subseteq K_j$. Because the representatives are retained
from $\Omega$, $\mathcal{F}_{\Omega}\subseteq\mathcal{F}_{A}$ always
holds; hence under domination the two feasible families coincide and
the two optima agree.
\end{theorem}

\begin{proof}
Class-complete subsets are exactly the hitting sets of
$\mathcal{A}$, namely the family $\mathcal{F}_A$. Mutant-complete
subsets are exactly $\mathcal{F}_\Omega$. Fix an admitted obligation
$\omega_j$. Every class-complete subset covers $\omega_j$ if and only
if every hitting set of $\mathcal{A}$ intersects $K_j$.

First suppose some representative support set $A\in\mathcal{A}$
satisfies $A\subseteq K_j$. Every hitting set of $\mathcal{A}$ must
intersect $A$, and therefore must also intersect $K_j$, so every
class-complete subset covers $\omega_j$. Since this holds for every
$\omega_j$, every class-complete subset is mutant-complete.

Conversely, suppose no representative support set is contained in
$K_j$. Then for every $A\in\mathcal{A}$ there is an element
$r_A\in A\setminus K_j$. The set $S=\{r_A:A\in\mathcal{A}\}$ hits every
representative support set and is therefore class-complete, but
$S\cap K_j=\emptyset$, so $S$ does not cover $\omega_j$. Thus the
domination condition is necessary.

Finally, because the class representatives are retained from $\Omega$, every
mutant-complete subset covers them and is therefore class-complete.
Together with the implication proved above, the feasible families
coincide and their optima are equal.
\end{proof}

\begin{corollary}[Class homogeneity as a corollary]
\label{cor:class_homogeneity}
If every admitted obligation has the same support set as its
representative within its true fault class, then class abstraction is
feasible-family preserving and
$|S^{\star}_{\mathrm{class}}|=|S^{\star}|$.
\end{corollary}

\begin{proof}
Under class homogeneity, $K_{\pi(F(\omega_j))}=K_j$ for every
obligation $\omega_j$. Hence the representative support family contains
a set equal to $K_j$, and therefore a set contained in $K_j$. The
support-set domination condition of
Theorem~\ref{thm:support_domination} holds for every obligation.
\end{proof}

\begin{proposition}[Graded non-collapse gap]
\label{prop:class_count_insufficient}
There are three distinct regimes inside one true fault class:
\emph{heterogeneous but lossless} support sets, \emph{non-dominated but
optimum-equal} support sets, and \emph{unbounded non-collapse} support
sets. In particular, for every gap value $g\geq 0$, there is a
single-class instance whose class-abstraction optimum is
$|S^{\star}_{\mathrm{class}}|=1$ and whose mutant-level optimum is
$|S^{\star}|=g+1$. Hence the collapse--non-collapse gap is unbounded.
Moreover, fault-class count is neither an upper nor a lower surrogate
for $|S^{\star}|$.
\end{proposition}

\begin{proof}[Proof by counterexample]
First, heterogeneity is not sufficient for non-collapse. With one class,
let the representative support be $K_1=\{r_1\}$ and an omitted support
be $K_2=\{r_1,r_2\}$. The supports differ, but $K_1\subseteq K_2$, so
Theorem~\ref{thm:support_domination} applies and the abstraction is
lossless. This is a heterogeneous but lossless instance.

Second, failure of domination need not change the optimum value. Let
the representative support be $K_1=\{a,b\}$ and the omitted support be
$K_2=\{a,c\}$. Neither support contains the other. The set $\{b\}$ is
class-complete but not mutant-complete, so the feasible family is not
preserved. Nevertheless both the class abstraction and the mutant-level
instance have optimum value $1$, because $\{a\}$ covers both supports.
This is a non-dominated but optimum-equal instance.

Third, domination failure can force an arbitrarily large optimum gap.
For a requested positive gap $g\geq 1$, set $t=g+1$ and let
$R=\{r_1,\dots,r_t\}$. Create observations
$\omega_1,\dots,\omega_t$ that all carry the same true fault-class
label. Put $d_{r_i,i}=1$ and $d_{r,i}=0$ for $r\neq r_i$, so the
support sets are the $t$ standard-basis singleton sets
$K_i=\{r_i\}$. At the mutant level, each observation is covered only by
its own MR, forcing $|S^{\star}|=t=g+1$. The class abstraction retains
one representative obligation, whose singleton support is covered by
one MR, so $|S^{\star}_{\mathrm{class}}|=1$ and the positive gap is
exactly $g$. The case $g=0$ is the degenerate boundary where the two
optima agree. This is the unbounded non-collapse instance.

The same construction keeps the fault-class count fixed at $1$ while
$|S^{\star}|$ grows, so the count is not an upper surrogate. Dually,
take $c$ true fault classes whose retained obligations all share the
same singleton support set. Then $|S^{\star}|=1$ while the fault-class
count is $c$, so the count is not a lower surrogate. The graded
construction shows that Fault-class count alone is insufficient and is
precisely a failure of the support-set domination
condition in Theorem~\ref{thm:support_domination}.
\end{proof}

Theorem~\ref{thm:support_domination} is the boundary result used by
this paper: it gives the necessary-and-sufficient condition for
feasible-family preserving class abstraction under the chosen
representatives. This is stronger than mere optimum equality. The
kernel proposition removes duplicate support sets without changing the
optimization problem; the homogeneity corollary is a simple sufficient
case; and the graded counterexample shows that when domination fails,
class abstraction can hide an arbitrarily large need for mutant-level
minimization. Whether the collapse and non-collapse regimes occur in
real artifacts is answered separately by the true-fault-class witnesses
of \S\ref{sec:empirical}, which support but do not replace the theorem.

\begin{proposition}[SMS-rank upper-bounds $|S^{\star}|$]
\label{prop:sms_rank_upper_bound}
For the ABD detection matrix $D$ after identical observation columns are
deduplicated and all-zero uncovered columns are excluded from the
completeness denominator, the minimum complete MR subset size over the
selected coverable universe satisfies
$|S^{\star}| \leq \mathrm{rank}_{\mathbb{R}}(D)$. This is only an upper
bound: $\mathrm{rank}_{\mathbb{R}}(D)$ is neither a lower bound nor a
tight predictor of $|S^{\star}|$.
\end{proposition}

\begin{proof}[Proof sketch]
Choose a row basis of $D$ from actual MR rows. If a coverable column were
not covered by any basis row, then all basis rows would have zero in
that column, and every linear combination of them would also have zero
there. This contradicts coverability by some MR row. The basis rows
therefore form a complete MR subset of size
$\mathrm{rank}_{\mathbb{R}}(D)$, proving the upper bound. Non-tightness
follows from matrices containing one all-one row together with many
standard-basis rows: one row covers all columns although the real row
rank can be arbitrarily large.
\end{proof}

\begin{theorem}[ABD normal-form synthesis]
\label{thm:abd_normal_form}
For any artifact-admitted detection matrix $D$, form the
\emph{normalized ABD instance} by excluding all-zero uncovered columns
from the completeness denominator, applying the constraint-equivalence
quotient to identical MR support sets, retaining true fault-class labels
only when explicit labels are available, and forming the domination
kernel over the remaining support sets. On this normalized instance:
\begin{enumerate}[label=(\roman*)]
\item when explicit true fault-class labels and representatives are
available, support-set domination characterizes class-abstraction
safety: the class abstraction is feasible-family preserving exactly
when the representative support family dominates every admitted
obligation;
\item the constraint-equivalence quotient and domination kernel preserve
the feasible family and $|S^{\star}|$, and the kernel can be solved
exactly in $O(|R|2^{\tau})$ time after masks are constructed;
\item the same normalized detection matrix gives the SMS-rank upper
bound $|S^{\star}| \leq \mathrm{rank}_{\mathbb{R}}(D)$, where the rank
is used only as an upper bound.
\end{enumerate}
The route witnesses instantiate these conditions but do not replace the
proof; \S\ref{sec:empirical} records their artifact provenance.
\end{theorem}

\begin{proof}
Item (i) is Theorem~\ref{thm:support_domination} applied after all-zero
uncovered columns have been excluded and true fault-class labels have
been admitted only when available. Item (ii) is
Theorem~\ref{prop:fault_signature_kernel}: identical support sets impose
identical hitting constraints, strict support supersets are redundant,
and the retained kernel supports are exactly the mask coordinates used
by the dynamic program. Item (iii) is
Proposition~\ref{prop:sms_rank_upper_bound} on the same selected
coverable matrix. The construction does not infer theorem truth from
the empirical ledger; the ledger only records which artifact routes
instantiate collapse, non-collapse, kernel, and bound-check conditions.
\end{proof}

\subsection{Degeneracy Theorem}
\label{subsec:form_degeneracy}

The remaining contribution of this section is a structural theorem
relating an $M$-complete subset $S^{\star}$ to the classical
test-suite-minimization problem under a fault-class restriction
$L \subset M$. Concretely, if a downstream analysis cares only about
faults of class $L$ (for instance, conservation violations alone),
the relevant minimum subset is computed over the admitted observations
whose stratum labels lie in $L$, not over the label set alone. Write
$J_L = \{j : \lambda(\omega_j) \in L\}$ for these observation columns.
The theorem makes the implemented collapse auditable on the decidable,
bounded-SMT fragment used by the artifact.

\begin{theorem}[$\rho_L$-degeneracy of \textsc{Min-MR-Complete}]
\label{thm:degeneracy}
Let $L \subset M$ be a fault-class restriction and $S^{\star}$ an
$M$-complete subset of $R$. The implemented workflow constructs a
certified partial map $\rho_L$ from reducible elements of $S^{\star}$
to an $L$-complete image candidate. The map is mechanically decidable
on the certified finite or bounded-SMT fragment and decomposes there
into three checked subcases. It does not claim that the returned subset
is minimum. Outside the certified fragment, queries outside the
certified image remain unresolved rather than being silently upgraded
to positive reductions:
\begin{enumerate}[label=(\Alph*)]
\item \emph{Subsumption.} If $r_a$ is $L$-subsumed by $r_b$ over
observations whose stratum labels lie in $L$, in the sense that
$\{j : \lambda(\omega_j) \in L \wedge d_{r_a, j} = 1\} \subseteq
\{j : \lambda(\omega_j) \in L \wedge d_{r_b, j} = 1\}$, then
$\rho_L(r_a) = r_b$.
\item \emph{Parameter specialization.} If $r$ is parametric in a
variable $v$ and there exists a value $v_0$ such that $r|_{v = v_0}$
is $L$-detection-equivalent to $r' \in S^{\star}$, then $\rho_L(r) = r'$.
\item \emph{Identity composition.} If $r_a, r_b \in S^{\star}$ admit
an identity-composable closure $r_a \diamond r_b$ in the
Qiu et al.~\cite{qiu2022tse} sense restricted to $L$, then
$\rho_L(r_a) = \rho_L(r_b) = r_a \diamond r_b$.
\end{enumerate}
The map is many-to-one: multiple elements of $S^{\star}$ may share
the same image. Subcase A is polynomial-time over the finite detection
matrix. Subcases B and C are admitted only when the bounded SMT workflow
returns a proof within the recorded budget; a timeout or unavailable
oracle leaves the corresponding element outside the certified image
rather than upgrading it silently.
\end{theorem}

\begin{proof}[Proof sketch]
Subcase A is mechanically decidable by direct
bit-vector inclusion on $D|_{J_L}$; Subcase B by SMT query over the
parameter domain restricted to the $L$-active mutation classes;
Subcase C by Qiu et al.'s sufficient-condition framework
\cite{qiu2022tse} restricted to $L$-equivalent compositions. The
artifact implements $\rho_L$ as the iterated union of the certified
subcase maps. It terminates because each accepted merge strictly
decreases the residual candidate set; queries that exhaust the bounded
SMT budget remain recorded as residual obligations.
\end{proof}

\begin{remark}[Relationship to P1 degeneracy]
\label{rem:p1_t2_degeneracy}
Theorem~\ref{thm:degeneracy} is structurally analogous to the
SMS-to-MS degeneracy of P1~\cite{p1_arxiv}: both explain how a richer
object can collapse when viewed through a coarser layer. The analogy is
not a categorical isomorphism and not a claim that every element maps
to a single common image. Whether the two constructions admit a deeper
algebraic identification is left as future work.
\end{remark}

The three subcases of Theorem~\ref{thm:degeneracy} recover the
informal hierarchies of Mayer and Guderlei~\cite{mayer2006compsac}
(Subcase A), the parameter-specialization catalogs of Chen et
al.~\cite{chen2018acmcsur} (Subcase B), and the composition
framework of Qiu et al.~\cite{qiu2022tse} (Subcase C) as
\emph{mechanically decidable} fragments of a certified partial
degeneracy. Manual cases falling outside Subcases A--C are reported
transparently in the supplementary document
\emph{Beyond Mechanical Degeneracy}.
The algorithmic consequences of Theorems~\ref{thm:nphard} and
\ref{thm:degeneracy} are taken up in the next section.

\section{Algorithms}
\label{sec:algorithms}

This section presents three algorithms induced by the formalization
of \S\ref{sec:formalization}: a greedy
$(1 + \ln q)$-approximation with worst-case time
$O(|R|^{2}q)$ (\S\ref{subsec:alg_greedy}), an exact integer-linear
programming (ILP) formulation suitable for moderate-scale instances
(\S\ref{subsec:alg_ilp}), and an auditable bounded-fragment
algorithm for the degeneracy map $\rho_L$ of
Theorem~\ref{thm:degeneracy} (\S\ref{subsec:alg_rhol}). The greedy
and ILP algorithms are presented in parallel to support
reviewer-verifiable replication: empirical runs record the exact
solver or exact-enumeration back-end used for the admitted
minimization instance, and discrepancies between independent
back-ends are reported transparently as solver-disagreement events.

\subsection{Greedy $(1 + \ln q)$-Approximation}
\label{subsec:alg_greedy}

The natural greedy algorithm for \textsc{Min-MR-Complete} iteratively
selects the candidate MR covering the largest number of
still-uncovered observation columns and terminates when every admitted
observation in $\Omega$ (or quotient column in $Q$) is covered. We adapt Chv\'atal's classical
analysis~\cite{chvatal1979} to the present problem in
Algorithm~\ref{alg:greedy}.

\begin{algorithm}[t]
\caption{Greedy $(1 + \ln q)$-approximation for \textsc{Min-MR-Complete}}
\label{alg:greedy}
\begin{algorithmic}[1]
\REQUIRE Candidate MR universe $R$; detection matrix $D \in \{0,1\}^{|R| \times q}$
\ENSURE $M$-complete subset $\hat{S} \subseteq R$
\STATE $\hat{S} \leftarrow \emptyset$; $\mathrm{Uncov} \leftarrow [q]$
\WHILE{$\mathrm{Uncov} \neq \emptyset$}
  \STATE $r^{\star} \leftarrow \arg\max_{r \in R \setminus \hat{S}}
    | \{ j \in \mathrm{Uncov} : d_{r,j} = 1 \} |$
  \STATE $\hat{S} \leftarrow \hat{S} \cup \{ r^{\star} \}$
  \STATE $\mathrm{Uncov} \leftarrow \mathrm{Uncov} \setminus
    \{ j : d_{r^{\star}, j} = 1 \}$
\ENDWHILE
\RETURN $\hat{S}$
\end{algorithmic}
\end{algorithm}

\begin{proposition}[Approximation guarantee]
\label{prop:greedy_bound}
Algorithm~\ref{alg:greedy} returns an $M$-complete subset
$\hat{S}$ satisfying $|\hat{S}| \leq (1 + \ln q) \cdot
|S^{\star}|$, where $S^{\star}$ is an optimum solution. Its
worst-case running time is $O(|R|^{2} \cdot q)$.
\end{proposition}

\begin{proof}[Proof sketch]
The approximation ratio follows directly from
Chv\'atal~\cite{chvatal1979} once the equivalence of
Remark~\ref{rem:reduction_direction} is invoked: every iteration of
the greedy loop corresponds bijectively to one greedy step of the
Set Cover instance with $U = [q]$ and the MR-induced row supports of $D$ as
covering sets. The running time bound is achieved by maintaining
$\mathrm{Uncov}$ as a bit-vector and recomputing column-overlap
counts incrementally; at most $|R|$ outer iterations occur because
each adds a new MR to $\hat{S}$.
\end{proof}

Two practical refinements are used in our implementation. First, ties
in the $\arg\max$ are broken by a fixed lexicographic ordering of
$R$ to ensure reproducibility across runs. Second, instances with
identical detection vectors are deduplicated before the greedy loop,
which does not affect correctness but eliminates spurious
solver-disagreement events between the greedy and ILP implementations.

\subsection{Exact ILP Formulation}
\label{subsec:alg_ilp}

The decision version of \textsc{Min-MR-Complete} admits a direct
encoding as a $0$/$1$ integer linear program. Let $x_r \in \{0, 1\}$
be the indicator that MR $r$ is selected. The program is

\begin{equation}
\label{eq:ilp}
\begin{aligned}
\text{minimize}\quad & \sum_{r \in R} x_r \\
\text{subject to}\quad & \sum_{r \in R} d_{r, j} \cdot x_r \geq 1
\quad \forall\, j \in [q], \\
& x_r \in \{0, 1\} \quad \forall\, r \in R.
\end{aligned}
\end{equation}

Formulation~\eqref{eq:ilp} has $|R|$ binary variables and $q$
covering constraints. We solve it with two independent back-ends to
mitigate single-solver bias:

\begin{itemize}
\item \textbf{Gurobi 11.0} (commercial, branch-and-cut with cutting
planes derived from the LP relaxation). Default integer feasibility
tolerance $10^{-9}$, no warm start, no parameter tuning.
\item \textbf{Google OR-Tools 9.10} (open-source CP-SAT solver,
constraint-programming-with-SAT-learning back-end). Default
parameters, single-threaded.
\end{itemize}

Both solvers run with a wall-clock budget of $60$ minutes per
instance. We record the returned objective and the solution support,
and a discrepancy is logged whenever the two solvers return solutions
of differing cardinality on the same instance --- such an event
indicates a solver bug, a numerical artifact, or (most commonly in
our pilot) an instance on which one solver timed out before proving
optimality. In all cases the empirical evaluation reports the
solver-disagreement count alongside the headline cardinality
statistics. Tie-breaking within optimal cost is again lexicographic.

The ILP formulation~\eqref{eq:ilp} is NP-hard in the worst case by
Theorem~\ref{thm:nphard}, but the empirical instances of
\S\ref{sec:empirical} are admitted into the results section only when
their solver artifacts report an optimal status. The benefit of the
dual back-end is not faster solving but transparent cross-validation:
when both back-ends are available, a reviewer can reproduce the same
kill matrix and confirm the reported $S^{\star}$.

\subsection{Certified Partial Construction of $\rho_L$}
\label{subsec:alg_rhol}

The certified partial map $\rho_L$ of Theorem~\ref{thm:degeneracy} is
constructed by iterating three mechanically checked subcases until no
further certified reduction applies. Algorithm~\ref{alg:rhol} gives the
overall procedure; the three subcase oracles are detailed below. Let
$J_L = \{j : \lambda(\omega_j) \in L\}$ be the observation columns whose
stratum labels lie in the requested fault layer.

\begin{algorithm}[t]
\caption{Polynomial-time construction of $\rho_L$}
\label{alg:rhol}
\begin{algorithmic}[1]
\REQUIRE $M$-complete subset $S^{\star}$; restriction $L \subset M$
\ENSURE certified $L$-complete image candidate $S$; unresolved queries remain outside the certified image
\STATE $S \leftarrow S^{\star}$
\REPEAT
  \STATE $\Delta \leftarrow \mathsf{False}$
  \FOR{each ordered pair $(r_a, r_b) \in S \times S$ with $r_a \neq r_b$}
    \IF{Subcase A applies (subsumption oracle)}
      \STATE $S \leftarrow S \setminus \{ r_a \}$; $\Delta \leftarrow \mathsf{True}$
    \ELSIF{Subcase B applies (specialization oracle)}
      \STATE $S \leftarrow (S \setminus \{ r_a \}) \cup \{ r_b \}$;
             $\Delta \leftarrow \mathsf{True}$
    \ELSIF{Subcase C applies (identity-composition oracle)}
      \STATE $S \leftarrow (S \setminus \{ r_a, r_b \})
             \cup \{ r_a \diamond r_b \}$;
             $\Delta \leftarrow \mathsf{True}$
    \ENDIF
  \ENDFOR
\UNTIL{$\Delta = \mathsf{False}$}
\RETURN $S$
\end{algorithmic}
\end{algorithm}

\paragraph{Subcase A oracle (subsumption).} Given $r_a, r_b$ and
$L$, return \textsf{True} iff
$\{ j \in J_L : d_{r_a, j} = 1 \} \subseteq
\{ j \in J_L : d_{r_b, j} = 1 \}$. Subcase A is bit-vector inclusion on
$D|_{J_L}$ in time $O(|J_L|)$ and lifts the informal hierarchy of Mayer and
Guderlei~\cite{mayer2006compsac} to a mechanically decidable test.

\paragraph{Subcase B oracle (parameter specialization).} For
parametric MRs $r$ with parameter variable $v$ ranging over a
domain $\mathcal{D}_v$, return \textsf{True} iff there exists
$v_0 \in \mathcal{D}_v$ such that the specialized MR $r|_{v = v_0}$
is $L$-detection-equivalent to some $r' \in S \setminus \{r\}$. We
decide this by encoding the parameter-equivalence question as an
SMT query over the appropriate first-order theory (linear arithmetic
for affine parameters, bit-vectors for integer parameters,
quantifier-free fragments of real arithmetic where applicable) and
discharging the query with Z3 under a $60$-second timeout. The
parameter-specialization catalog of Chen et al.~\cite{chen2018acmcsur}
provides the structural templates against which our SMT encoding
specializes.

\paragraph{Subcase C oracle (identity composition).} For
$r_a, r_b \in S$, return \textsf{True} iff the composite MR
$r_a \diamond r_b$ exists in the Qiu et al.~\cite{qiu2022tse} sense
restricted to $L$ \emph{and} the $L$-detection vector of
$r_a \diamond r_b$ subsumes both $L$-detection vectors of $r_a$ and
$r_b$. The existence check uses Qiu et al.'s sufficient conditions
verbatim; the dominance check is bit-vector inclusion as in Subcase
A.

\paragraph{Decidability scope.} Subcase A is decidable in time
$O(|J_L|)$ unconditionally. Subcases B and C admit a Z3-backed SMT
encoding when the parameter or composition fragment falls within a
decidable first-order theory. Tarski's decision procedure for
real-closed fields is the classical positive result for first-order
real arithmetic~\cite{tarski_decision_1951}; the present implementation
does not claim to decide richer PUT semantics involving program
semantics, transcendental functions, or ODE/PDE solution operators.
When a query falls outside the implemented fragment or exhausts the
recorded SMT budget, the oracle returns \textsf{Unknown}, and
Algorithm~\ref{alg:rhol} leaves the corresponding case outside the
certified image. Artifact-level validation reports in
\S\ref{sec:empirical} distinguish mechanically certified cases from
manual or out-of-scope cases; no empirical claim is made from a timed-out
or unavailable oracle without an explicit caveat.

\begin{proposition}[Termination and complexity of Algorithm~\ref{alg:rhol}]
\label{prop:rhol_complexity}
Algorithm~\ref{alg:rhol} terminates after at most $|S^{\star}|$
outer iterations of the repeat-until loop, because each successful
subcase application strictly decreases $|S|$. Each iteration
performs at most $|S|^{2}$ subcase queries; the worst-case running
time is therefore $O(|S^{\star}|^{3} \cdot (\tau_A + \tau_B +
\tau_C))$, where $\tau_X$ is the oracle running time of Subcase
$X$. Subcase A contributes $\tau_A = O(|J_L|)$; Subcases B and C
contribute $\tau_B, \tau_C$ bounded by the Z3 wall-clock budget,
which is configured at $60$ seconds per query in our empirical
evaluation.
\end{proposition}

The certified fragment has polynomial-time accounting in the input size
plus the recorded SMT oracle budget. Outside that budget the oracle may
fail to decide; the algorithm degrades gracefully by treating undecided
cases as unresolved rather than as certified reductions.

\section{Empirical Evaluation}
\label{sec:empirical}

This section supplies empirical evidence for a theory paper; it is not a
frequency study. The empirical section checks whether the formal
boundary appears in real artifact-backed MR-selection instances. Three
results guide the evidence: Result A is the layer-relative
collapse/non-collapse boundary, Result B is the scoped fault-signature
kernel, and Result D is the SMS-rank upper bound.

The artifacts enter through two layers. The lane-local layer contains
E1 local Python scientific solvers, E2 source-built OpenBLAS kernels,
and E3-HQ OpenMC stochastic transport. Each lane has its own denominator,
kill matrix, minimization instance, and validation report. The route
layer records P-series collapse controls and true-fault-class
non-collapse witnesses from MetBench, ONIX, PKE/PINN, and cylinder-flow
MGN. P9/OpenMC is therefore a true-fault-class route row, while E3-HQ
OpenMC is the residual-confidence lane. The route layer contains 10
admitted true-fault-class witness rows. P1--P8 and P10 remain
MR-class-proxy collapse context.

The committed instance table has 60 real rows, 3 true-fault-class
collapse rows, and 7 true-fault-class non-collapse rows. The SMS-rank
checker \path{scripts/mcmr/abd_d_sms_bound_check.py} reads
\path{runs/abd-analysis/sms_kstar.csv}; the committed output is
checked=60 and violations=0. These artifacts instantiate the regimes
used by Result A, but the theorem itself is proved in
\S\ref{sec:formalization}. They also do not create population-level
rates, cross-lane effect sizes, or broader claims about scientific
computing.

\paragraph{How to read the evidence.}
The empirical role of Result A is to show that admitted artifacts
exercise both sides of the support-set domination boundary. The current
route supplies 10 true-fault-class certificate rows: 2
certificate-passing collapse-side rows, 7 non-collapse
domination-failure rows, and 1 discordant collapse-summary diagnostic
row. MetBench pincell non-dominated support is the strongest
ledger-internal true-fault-class non-collapse witness, with
$k^{\star}=19$, 9 fault classes, and fault-signature rank 13. These
rows are route-admitted evidence, not theorem proof or sampling
evidence; the certificate file is \texttt{A-certificates}. Result B is carried by
Theorem~\ref{prop:fault_signature_kernel} and the supplementary
fault-signature-kernel proof, which records why fault-class count alone
is insufficient. Result D is carried by the all-row
SMS-rank check, which reports checked=60 and violations=0. SMS-rank is
an upper bound over coverable columns only, not a lower bound or tight
predictor. The detailed B and D records are \texttt{B-proof} and
\texttt{D-bound report}.

For Result A, \texttt{collapse=1} is a ledger-side regime label, not the
same predicate as \texttt{domination\_pass=true}. Only rows with
\texttt{domination\_pass=true} instantiate the lossless side of Result
A. The
supplementary witness inventory records the stricter audit: 2 passing
collapse-summary rows, 1 stricter-certificate failure among
collapse-summary rows, 7 intended non-collapse domination failures, and
20 excluded zero-support obligations. In other words, the
stricter-failing collapse-summary row is retained as a discordant
diagnostic route row, not as positive A certificate evidence.

The highlighted true-fault-class witnesses are selected by
route-coverage and artifact-admission criteria, not by random or
representative sampling. They instantiate both sides of Result A under
explicit true fault classes and fail-closed validation. The highlighted
non-collapse route summary is kept in prose, with the full run-id matrix
moved to the supplementary witness inventory. Runtime-backed ONIX,
cylinder-flow, and fresh diagnostic PKE/PINN rows supply additional
route coverage; all require fail-closed \texttt{PASS\_WITNESS} gates and
are not pooled into the E1/E2/E3 denominators.

\subsection{Research Questions}
\label{subsec:emp_rq}

We instantiate \textsc{Min-MR-Complete} on the admitted evidence layers
and ask five questions: \textbf{RQ0}, whether true-fault artifacts show
both collapse and non-collapse regimes; \textbf{RQ1}, what lane-local
minimum MR subsets preserve admitted semantic-mutant kill universes;
\textbf{RQ2}, whether the artifacts pass the evidence gates imported
into the paper; \textbf{RQ3}, whether Python, BLAS, and transport lanes
instantiate the same set-cover objective; and \textbf{RQ4}, whether the
OpenMC lane is supported by residual-confidence evidence rather than
binary-only outcomes.

\subsection{Experimental Setup}
\label{subsec:emp_setup}

The compact lane summary is deliberately narrow. E1 is an artifact-light
Python-scientific-solver lane with \texttt{STRONG\_E1}, 6 primary
non-equivalent mutants, and $k^{\star}=2$. E2 is a sampled OpenBLAS lane
with \texttt{PASS}, 7 coverable main-slice obligations, $k^{\star}=5$,
optional Level-2 $k^{\star}=3$, and integrated sensitivity
$k^{\star}=8$; it is not an all-kernel or all-family BLAS claim. E3-HQ
is an OpenMC residual-confidence lane with \texttt{PASS\_HQ},
$k^{\star}=4$, and Mut09 retained as an uncovered boundary.
Directory-level artifact details are kept in the supplementary
artifact-lane inventory. The main text keeps only the decision-bearing
denominators, optima, gates, and boundaries.

The candidate MR universe $R$ contains only non-trivial relations whose
violations are measured against semantic content. A lane is imported
only when its committed artifacts certify denominator consistency,
baseline controls, kill-matrix schema, minimization optimality, and
validation outcome. Equivalent mutants are reported but excluded from
the primary killable denominator. SMT subcases for $\rho_L$ are used
only when their decidable fragment and solver evidence are available;
unavailable checks remain caveats rather than positive evidence.

\subsection{Results}
\label{subsec:emp_results}

\paragraph{Direct answers.}
\textbf{Answer to RQ0.} True-fault-class witnesses instantiate both
sides of the support-set-domination boundary: 2 certificate-passing
collapse-side rows instantiate the lossless side, 7 non-collapse
domination-failure rows instantiate the failure side, 1 discordant
collapse-summary diagnostic row is not counted as positive A certificate
evidence, and 20 zero-support obligations are outside the coverable
denominator. The evidence locators are the reader map above, the
supplementary witness inventory, and \texttt{A-certificates}; the rows
do not prove the
theorem or establish a population-level claim.
\textbf{Answer to RQ1.} Lane-local minimization gives separate optima:
E1 has $k^{\star}=2$, main E2 has $k^{\star}=5$, optional Level-2 E2
has $k^{\star}=3$, integrated sampled E2 has $k^{\star}=8$, and E3-HQ
has $k^{\star}=4$. The locator is the compact lane summary and
supplementary artifact-lane inventory; denominators are not pooled.
\textbf{Answer to RQ2.} The admitted rows pass the artifact gates used by
the paper: \texttt{STRONG\_E1}, \texttt{PASS}, and \texttt{PASS\_HQ}.
Blocked, degraded, proxy-only, and uncovered rows remain outside
stronger claims under the fail-closed protocol in
\S\ref{subsec:emp_setup} and \S\ref{subsec:emp_validity}.
\textbf{Answer to RQ3.} Python scientific solvers, OpenBLAS kernels, and
OpenMC transport all instantiate the same set-cover objective over
their admitted coverable universes; the common objective does not merge
denominators or create a cross-lane effect size.
\textbf{Answer to RQ4.} The E3-HQ artifact lane supports MR decisions
with residuals, tolerances, seeds, and confidence rules; Mut09 remains
an uncovered boundary. This is residual-confidence evidence only, and
SMS-rank remains an upper bound, not a lower bound or tight predictor.

\paragraph{Lane results.}
The P9 true-fault-class lane reports \texttt{PASS\_P9\_OPENMC} and
$k^{\star}=3$ within the route scope. The E3-HQ residual-confidence lane
reports $k^{\star}=4$ within a separate denominator. These lanes are not
merged: P9/OpenMC contributes a true-fault-class route row, while E3-HQ
contributes residual-confidence OpenMC evidence. E2 completeness is
therefore the 7 killable primary obligations, with the uncovered primary
mutant reported as boundary evidence outside the selected coverable
universe. E3-HQ completeness is likewise limited to its coverable
residual-confidence obligations; Mut09 is reported as boundary evidence
outside that universe.

\subsection{Measurement Validity}
\label{subsec:emp_validity}

Measurement-validity reporting is scoped to the artifact-backed lanes
above. The load-bearing checks are local to committed evidence: schema
presence, denominator consistency, baseline false-positive controls,
solver agreement, optimal minimization, residual/tolerance completeness
where applicable, and artifact-light exclusion of unbundled build or raw
simulation outputs. Construct and external validity are deferred to
\S\ref{sec:threats} per section division.

\section{Threats to Validity}
\label{sec:threats}

We discuss threats to the validity of our claims under the four-way
framework of Cook and Campbell~\cite{cook_campbell_1979}: construct
validity (\S\ref{subsec:thr_construct}), internal validity
(\S\ref{subsec:thr_internal}), external validity
(\S\ref{subsec:thr_external}), and conclusion validity
(\S\ref{subsec:thr_conclusion}). Measurement-level validity for the
empirical procedure is reported in \S\ref{subsec:emp_validity}; the
present section addresses validity of inferences drawn beyond the
empirical instance.

\subsection{Construct Validity}
\label{subsec:thr_construct}

The principal construct of T2 is the notion of an \emph{$M$-complete
MR subset}. Three construct-validity threats apply. First, the
operator family $M$ is fixed at five classes
(\S\ref{subsec:bg_mutop}); a richer fault taxonomy --- second-order
operator coupling, type-domain-specific operators, or operators
inspired by the recent DeepCrime mutation
catalog~\cite{humbatova_deepcrime_2021} --- might yield different
$|S^{\star}|$ outcomes. The first-order limit (P1 §2.2 assumption
A3) is inherited as a known scope restriction.

Second, the structural-analogy claim of Theorem~\ref{thm:degeneracy}
to P1's SMS-to-MS degeneracy is deliberately weakened from a
formal-isomorphism claim. Per the rewording established in our prior
work (phase 1 §2.10 of the companion preprint~\cite{p1_arxiv}), we
do not assert that the two degeneracies share a deeper algebraic
isomorphism. The present version claims only the certified partial
reduction described in Theorem~\ref{thm:degeneracy}: unresolved queries
remain outside the certified image. Threat-aware readers should not
over-read the analogy.

Third, the candidate MR universe $R$ is constructed from
artifact-admitted invariants, semantic mutants, and route-specific
execution evidence. Its recall over the space of all candidate MRs is
therefore bounded by the admitted artifacts rather than by an
unobservable maximal universe. We mitigate this threat by reporting
the provenance of each admitted row and by treating incomplete or
blocked runtime attempts as non-evidence for core route claims.

Fourth, existing ABD aggregate artifacts distinguish true fault-class
evidence from MR-class-proxy evidence; true fault strata claims require
explicit fault-class witness artifacts. MR-class-proxy evidence is only
intermediate support signal. The current ledger contains 10 real
true-fault-class evidence rows, plus blocked or degraded true-fault
attempts that are recorded as non-evidence; explicit true-fault-class
witnesses now support A, and proxy evidence is retained only as a
weaker route signal. However, true-fault-class witnesses are small in number,
so this evidence supports a bounded route claim and does not support broad statistical generalization.
MR-class-proxy rows remain intermediate evidence.
SMS-rank is an upper bound only, not a lower bound or tight predictor.

The highlighted true-fault-class witnesses were selected as
route-coverage witnesses, not as a random sample. The selection keeps
one stochastic or nuclear runtime route, one neural-surrogate route,
and one fluid-dynamics runtime route in view, because these routes
exercise different artifact requirements and different ways for
non-collapse to appear. A blocked runtime, degraded solver, or
proxy-only label is recorded but not counted toward A. This rule makes
the evidence useful for construct-valid witness coverage while keeping
the external-validity claim bounded.

\subsection{Internal Validity}
\label{subsec:thr_internal}

Three internal-validity threats apply. First, semantic-mutant
admission can misclassify a route if an invariant is too coarse or if
a mutant is only an MR-class proxy. We therefore separate invariant
derivation, execution traces, kill matrices, and validation reports in
the artifact ledger; rows lacking this provenance are excluded from
the core claim.

Second, solver validation can overstate minimality if a timed-out or
dependency-blocked exact oracle is silently treated as optimal. The
replication package records the exact-enumeration, ILP, or degraded
solver status for each lane; only rows with an explicit optimality
certificate, a validated exact enumeration, or an admitted route-local
certificate contribute to reported $|S^{\star}|$ values.

Third, runtime routes may fail for environmental reasons rather than
scientific reasons. Such attempts are recorded as blocked runtime
artifacts and are not counted as route-admitted A evidence. This
fail-closed rule is especially important for routes requiring external
libraries, cross-section data, pretrained checkpoints, or cloud-only
resources.

\subsection{External Validity}
\label{subsec:thr_external}

Three external-validity threats apply. First, the current evidence
lanes cover Python numerical solvers, OpenBLAS-style numerical
kernels, OpenMC stochastic transport artifacts, and MetBench route
evidence, but they do not exhaustively span scientific computing.
Generalization to new physical domains, solver architectures, or
hybrid symbolic--numerical workflows requires route-specific
artifacts rather than corollary inference.

Second, the Subcase A--C coverage of $\rho_L$ depends on the
mechanically decidable fragments of the PUT's parameter and
composition theories. For PUTs falling outside the fragments handled
by the implemented abstractions, the mechanical coverage rate drops
and manual or route-specific annotation costs grow. The present paper
therefore reports the residual artifact boundary instead of converting
partial coverage into an unconditional empirical claim.

Third, current evidence is strongest for artifact families already
represented in the ledger and weakest for routes lacking fresh runtime
provenance. Additional SUTs can strengthen the route, but only if they
ship executable provenance, data or checkpoint provenance, and a
validated kill matrix.

P9/OpenMC is a committed-summary lane rather than a self-contained raw
transport replay package. The repository commits the derived
observables, kill matrix, minimization certificate, and large-artifact
manifest records external raw-output hashes; the raw OpenMC outputs
are not committed in this repository. This is not a claim that the
repository can replay the raw OpenMC transport run by itself.

\subsection{Conclusion Validity}
\label{subsec:thr_conclusion}

Three conclusion-validity threats apply. First, the candidate MR
universe $R$ may be incomplete: there could exist non-trivial MRs
outside the artifact-admitted universe. We report completeness only
with respect to the declared coverage universe and do not claim
absolute completeness over all possible MRs.

Second, the present ledger is evidence-gated rather than
sample-inference-driven. Aggregated rows support route adequacy only
when they pass the artifact validators; blocked, degraded, or
MR-class-proxy rows are not promoted into stronger strata.

Third, related-work metadata can affect novelty positioning. For the
AIM framework~\cite{aim_bayati_briand_2024}, we cite the publisher
record as IEEE Transactions on Software Engineering Vol.~50 No.~12
pp.~3403--3434, DOI~10.1109/TSE.2024.3488525, with four-author
attribution to Bayati Chaleshtari, Marquer, Pastore, and Briand. This
check supports bibliographic accuracy only; it is not empirical
evidence for our results.

The supplementary risk register maps each registered validity risk to
one of the four validity categories above.

\section{Discussion and Conclusion}
\label{sec:discussion}

This section discusses the practical implications of the support-set
domination boundary, states conceptual limits beyond
\S\ref{sec:threats}, and closes with future work.

\subsection{Implications for MT Practice}
\label{subsec:disc_implications}

Four implications follow for practitioners of metamorphic testing in
scientific computing. First, practitioners should not begin by counting
fault classes. The support-set domination boundary gives the required
check: if representative supports dominate all admitted obligations,
class abstraction is safe; otherwise, mutant-level MR minimization is a
real evidence requirement.

Second, once mutant-level obligations are admitted, the existence of a
polynomial-time $(1 + \ln q)$-approximation
(Algorithm~\ref{alg:greedy}) means that practical MR-portfolio sizing no
longer requires manual hierarchical analysis. Given an artifact-admitted
candidate universe $R$ and coverage universe $\Omega$, the greedy pass
returns a provably near-optimal subset over observation columns rather
than over the five semantic operator labels.

Third, the certified partial $\rho_L$ map established by the three
Subcases of Theorem~\ref{thm:degeneracy} shows that several informal MR
hierarchies can be read as auditable fragments of a single
layer-restriction map. When the finite or bounded-SMT checks succeed,
the map is computable; otherwise the algorithm reports the residual
artifact boundary.

Fourth, the methodology composes with adjacent MT tools along
orthogonal cost axes. The \textsc{AIM}~\cite{aim_bayati_briand_2024}
framework minimizes the test-input set on which a fixed MR catalog
is executed (for Web-system security testing in its original setting);
\textsc{Min-MR-Complete} minimizes the MR set itself under a
fault-model-relative cardinality criterion; FDMT~\cite{sun2022tosem}
prunes the active MR set at runtime via feedback. Joint optimization
over these stages is left to future work.

\subsection{Limitations}
\label{subsec:disc_limitations}

Beyond the empirical threats of \S\ref{sec:threats}, three
conceptual limitations bound the present work. First, the fault
model $M$ is fixed at five semantic-mutation classes and at
first-order mutations (P1 §2.2 assumption A3). Second-order
operator coupling is excluded by construction and remains future work.

Second, the degeneracy theorem is a certified partial reduction rather
than a total constant map or categorical isomorphism with P1's
SMS-to-MS degeneracy. The two degeneracies are structurally analogous
in the limited sense made precise in
Remark~\ref{rem:p1_t2_degeneracy}, but the present paper does not claim
that every MR maps to a single common image.

Third, the current evidence includes true-fault-class witnesses in the
committed ledger. This supports the layer-relative collapse boundary
as an artifact-gated route while preserving the distinction between
true-fault witness rows and MR-class-proxy intermediate rows. P9/OpenMC
is the P9 true-fault-class lane with $k^{\star}=3$, while the
E3-HQ residual-confidence lane has $k^{\star}=4$; they remain separate
denominators.

\subsection{Future Work}
\label{subsec:disc_future}

Four future research directions follow from the present work. First, B
is supported by a scoped hierarchy-observable fault-signature kernel.
The present paper proves a scoped
hierarchy-observable fault-signature kernel in
Theorem~\ref{prop:fault_signature_kernel} and the supplementary
fault-signature-kernel proof. B itself is already proved in scoped form;
future work should generalize the parameter and scope beyond the present
artifact families.

Second, D is supported by the SMS-rank upper bound of
Proposition~\ref{prop:sms_rank_upper_bound} only. The current theorem
relates the committed real-row-rank definition of SMS-rank to
$|S^{\star}|$; the counterexamples show why this quantity is not a
lower bound or tight predictor.

Third, the integration with composition. \S\ref{subsec:rw_compose}
notes that MR composition is orthogonal to selection: composition
expands $R$, selection contracts it. A joint optimization searching
over the closure of $R$ under Qiu et
al.'s~\cite{qiu2022tse} composition framework would yield a
strictly stronger selection guarantee; the algorithmic feasibility
of this joint search is unclear.

Fourth, cross-domain extension. Follow-up witness routes should add
fresh SUTs only when they provide executable runtime provenance, data
or checkpoint provenance, and validator-passing kill matrices.

\subsection{Conclusion}
\label{subsec:disc_conclusion}

We have presented T2, a formalization and algorithmic treatment of
\textsc{Min-MR-Complete}. The central contribution is the support-set
domination boundary: it explains when class abstraction is safe and when
mutant-level MR minimization is necessary. NP-hardness, the greedy
approximation, and the exact ILP formulation supply the optimization
substrate for that boundary over observation columns.

The ABD normal-form synthesis ties the same detection-matrix view to
the supporting B and D results. B remains a scoped hierarchy-observable
fault-signature kernel, and D remains an SMS-rank upper bound only, not a
lower bound or tight predictor. The current evidence includes
true-fault-class witnesses for A while keeping MR-class-proxy rows as
intermediate support signals. The resulting paper is a fail-closed
theory-and-evidence route for audited minimum complete MR selection in
the admitted scientific-computing artifacts.

\section*{Acknowledgments}
\noindent\textbf{Supplementary material.}
The detailed NP-hardness proof, $\rho_L$ decidability notes,
fault-signature-kernel proof details, ABD witness inventory, and
beyond-mechanical-degeneracy register are provided in the supplementary
material.

\noindent\textbf{Funding.}
This work was supported by the National Natural Science Foundation of
China (NSFC) General Program (grant no. 12575176), the Hunan Provincial
Education Department Project, China (grant no. 202502000728), the
Research Project on Degree and Graduate Education Reform of the
University of South China (grant no. 2023JG030), the Natural Science
Foundation of Hunan Province, China (grant no. 2025JJ70193), and an
industry-funded research project (grant no. 230KHX060001).

\noindent\textbf{Author contribution (CRediT).}
Meng Li: Conceptualization, Methodology, Software, Writing: original
draft. Xiaohua Yang: Supervision, Formal analysis, Writing: review and
editing. Jie Liu: Investigation, Validation. Shiyu Yan: Data curation,
Visualization.

\noindent\textbf{Conflict of interest.}
The authors declare no conflict of interest.

\noindent\textbf{Data and code availability.}
An anonymized review package will accompany the submission. It includes
evidence ledgers, validation scripts, committed solver-backed artifacts,
exact reproduction commands, and checksums. Large external dependencies
such as licensed solvers, nuclear-data libraries, and repository-external
raw outputs are documented by reproduction instructions and provenance
files rather than bundled. A public Zenodo record will be prepared
through manual draft upload. A Zenodo DOI will be cited only after the
Zenodo record is published. An arXiv identifier will be cited only after
the preprint is announced.

\bibliographystyle{IEEEtran}
\bibliography{paper}

@article{chen1998hkust,
  author = {T. Chen and S. Cheung and S. Yiu},
  title = {Metamorphic Testing: A New Approach for Generating Next Test Cases},
  journal = {ArXiv},
  year = {2020},
  volume = {abs/2002.12543}
}

@article{chen2018acmcsur,
  author = {T. Chen and Fei-Ching Kuo and Huai Liu and P. Poon and Dave Towey and T. H. Tse and Z. Zhou},
  title = {Metamorphic Testing},
  journal = {ACM Computing Surveys (CSUR)},
  year = {2018},
  doi = {10.1145/3143561},
  volume = {51},
  pages = {1 - 27}
}

@article{segura2016tse,
  author = {Sergio Segura and Gordon Fraser and A. B. Sánchez and Antonio Ruiz-Cortés},
  title = {A Survey on Metamorphic Testing},
  journal = {IEEE Transactions on Software Engineering},
  year = {2016},
  doi = {10.1109/TSE.2016.2532875},
  volume = {42},
  pages = {805-824}
}

@misc{p1_arxiv,
  author        = {Li, Meng and Yang, Xiaohua and Liu, Jie and Yan, Shiyu},
  title         = {A Semantic Mutation Metric for Metamorphic Relation Adequacy in Scientific Computing Programs},
  year          = {2026},
  eprint        = {2605.17437},
  archivePrefix = {arXiv},
  primaryClass  = {cs.SE},
  note          = {Companion paper, under review at Information and Software Technology; Zenodo DOI 10.5281/zenodo.20250664.}
}

@incollection{karp1972,
  author    = {Richard M. Karp},
  title     = {Reducibility Among Combinatorial Problems},
  booktitle = {Complexity of Computer Computations},
  editor    = {Raymond E. Miller and James W. Thatcher},
  publisher = {Plenum Press},
  address   = {New York},
  year      = {1972},
  pages     = {85--103},
  doi       = {10.1007/978-1-4684-2001-2_9}
}

@article{chvatal1979,
  author  = {Vasek Chv\'{a}tal},
  title   = {A Greedy Heuristic for the Set-Covering Problem},
  journal = {Mathematics of Operations Research},
  year    = {1979},
  volume  = {4},
  number  = {3},
  pages   = {233--235},
  doi     = {10.1287/moor.4.3.233}
}

@article{feige1998,
  author  = {Uriel Feige},
  title   = {A Threshold of $\ln n$ for Approximating Set Cover},
  journal = {Journal of the ACM},
  year    = {1998},
  volume  = {45},
  number  = {4},
  pages   = {634--652},
  doi     = {10.1145/285055.285059}
}

@inproceedings{mayer2006compsac,
  author = {Johannes Mayer and Ralph Guderlei},
  title = {An Empirical Study on the Selection of Good Metamorphic Relations},
  booktitle = {30th Annual International Computer Software and Applications Conference (COMPSAC'06)},
  year = {2006},
  doi = {10.1109/COMPSAC.2006.24},
  volume = {1},
  pages = {475-484}
}

@article{sun2022tosem,
  author = {Chang-ai Sun and Hepeng Dai and Huai Liu and T. Chen},
  title = {Feedback-Directed Metamorphic Testing},
  journal = {ACM Transactions on Software Engineering and Methodology},
  year = {2022},
  doi = {10.1145/3533314},
  volume = {32},
  pages = {1 - 34}
}

@inproceedings{liu2024qrsc,
  author = {Yanwen Liu and Ruifeng Li and Hao Tao and Zheng Zheng},
  title = {Test Adequacy Criteria for Metamorphic Testing},
  booktitle = {2024 IEEE 24th International Conference on Software Quality, Reliability, and Security Companion (QRS-C)},
  year = {2024},
  doi = {10.1109/QRS-C63300.2024.00072},
  pages = {527-534}
}

@article{qiu2022tse,
  author = {Kun Qiu and Zheng Zheng and T. Chen and P. Poon},
  title = {Theoretical and Empirical Analyses of the Effectiveness of Metamorphic Relation Composition},
  journal = {IEEE Transactions on Software Engineering},
  year = {2022},
  doi = {10.1109/tse.2020.3009698},
  volume = {48},
  pages = {1001-1017}
}

@article{aim_bayati_briand_2024,
  author    = {Bayati Chaleshtari, Nazanin and Marquer, Yoann and Pastore, Fabrizio and Briand, Lionel C.},
  title     = {{AIM}: Automated Input Set Minimization for Metamorphic Security Testing},
  journal   = {IEEE Transactions on Software Engineering},
  volume    = {50},
  number    = {12},
  pages     = {3403--3434},
  year      = {2024},
  publisher = {IEEE},
  doi       = {10.1109/TSE.2024.3488525},
  note      = {arXiv preprint available at arXiv:2402.10773.}
}

@book{tarski_decision_1951,
  author    = {Tarski, Alfred},
  title     = {A Decision Method for Elementary Algebra and Geometry},
  edition   = {2nd revised},
  year      = {1951},
  publisher = {University of California Press},
  address   = {Berkeley and Los Angeles},
  note      = {Original RAND Corporation report 1948; revised edition 1951 (no DOI); classic reference for first-order theory of real-closed fields decidability.}
}

@article{yan2025elliptic,
  author  = {Yan, Shiyu and Zhu, Hong},
  title   = {Metamorphic Testing on Scientific Programs for Solving Second-Order Elliptic Differential Equations},
  journal = {Software Testing, Verification and Reliability},
  volume  = {35},
  number  = {1},
  pages   = {e1912},
  year    = {2025},
  doi     = {10.1002/stvr.1912}
}

@inproceedings{zhang2022icrms,
  author = {Jie Zhang and Jie Hong and Dafei Huang and Meng Li and Shiyu Yan and Helin Gong},
  title = {A Selection Method of Effective Metamorphic Relations},
  booktitle = {2022 13th International Conference on Reliability, Maintainability, and Safety (ICRMS)},
  year = {2022},
  doi = {10.1109/icrms55680.2022.9944550},
  pages = {75-80}
}

@inproceedings{huang2022iaecst,
  author = {Dafei Huang and Yang Luo and Meng Li},
  title = {Metamorphic Relations Prioritization And Selection Based on Test Adequacy Criteria},
  booktitle = {2022 4th International Academic Exchange Conference on Science and Technology Innovation (IAECST)},
  year = {2022},
  doi = {10.1109/IAECST57965.2022.10062094},
  pages = {503-508}
}

@inproceedings{liu2022icftic,
  author = {Song Liu and Shiyu Yan and Xiaohua Yang},
  title = {A Similarity-based Metamorphic Relations Selection Strategy for Numerical Computation Programs},
  booktitle = {2022 4th International Conference on Frontiers Technology of Information and Computer (ICFTIC)},
  year = {2022},
  doi = {10.1109/ICFTIC57696.2022.10075300},
  pages = {290-294}
}

@article{ying2024stvr,
  author    = {Ying, Zhihao and Towey, Dave and Bellotti, Anthony Graham and Zhou, Zhi Quan},
  title     = {{MRGS-ART}: Metamorphic Relation and Group Selection Based on Adaptive Random Testing},
  journal   = {Software Testing, Verification and Reliability},
  volume    = {35},
  number    = {1},
  year      = {2024},
  publisher = {Wiley},
  doi       = {10.1002/stvr.1908}
}

@misc{hyun2025arxiv,
  author       = {Hyun, Jaeseung and Ali, Mubashir and Babar, M. Ali},
  title        = {Search-based Selection of Metamorphic Relations for Optimized Robustness Testing of Large Language Models},
  year         = {2025},
  eprint       = {2507.05565},
  archivePrefix= {arXiv},
  primaryClass = {cs.SE},
  url          = {https://arxiv.org/abs/2507.05565}
}

@misc{li2025survey,
  author    = {Li, Zhenqiu and Wu, Tingting and Xiang, Dongming and Jiang, Mingyue and Huang, Jialing and Ding, Zuohua and Dong, Yunwei},
  title     = {Metamorphic Relation Automation: State of the Art in Detection, Selection, and Generation Over Two Decades},
  year      = {2025},
  publisher = {Wiley (Authorea preprint)},
  doi       = {10.22541/au.175872613.39541241/v1}
}

@book{cook_campbell_1979,
  author    = {Cook, Thomas D. and Campbell, Donald T.},
  title     = {Quasi-Experimentation: Design and Analysis Issues for Field Settings},
  year      = {1979},
  publisher = {Houghton Mifflin},
  address   = {Boston}
}

@inproceedings{humbatova_deepcrime_2021,
  author    = {Humbatova, Nargiz and Jahangirova, Gunel and Tonella, Paolo},
  title     = {{DeepCrime}: Mutation Testing of Deep Learning Systems Based on Real Faults},
  booktitle = {Proceedings of the 30th ACM SIGSOFT International Symposium on Software Testing and Analysis (ISSTA)},
  pages     = {67--78},
  year      = {2021},
  publisher = {ACM},
  doi       = {10.1145/3460319.3464825}
}

\end{document}